\documentclass[reqno]{amsart}
\usepackage{graphicx,graphics}
\usepackage{subfig}

\newcommand{\D}{\mathcal{D}}

\newcommand{\Real}{\mathbb R}
\newcommand{\abs}[1]{\left\vert#1\right\vert}

\theoremstyle{definition}

\newtheorem{lem}{Lemma}
\newtheorem{thm}{Theorem}

\newtheorem{rem}{Remark}
\newtheorem{defn}{Definition}

\numberwithin{thm}{section}
\numberwithin{lem}{section}
\numberwithin{coll}{section}
\numberwithin{rem}{section}
\numberwithin{exm}{section}
\numberwithin{prop}{section}
\numberwithin{equation}{section}

\numberwithin{equation}{section}

\begin{document}
\centerline {\textsc {\large On R\'{e}nyi entropy convergence of the max domain of attraction}} 
\vspace{0.5in}
\begin{center}
   Ali Saeb\footnote{Corresponding author: ali.saeb@gmail.com}\\
Theoretical Statistics and Mathematics Unit,\\ Indian Statistical Institute, Delhi Center,\\7 S.J.S Sansanwal Marg, New Delhi 110016, India
\end{center}
\vspace{1in}

\noindent {\bf Abstract:}
In this paper, we prove that the R\'{e}nyi entropy of linearly normalized partial maxima of independent and identically distributed random variables is convergent to the corresponding limit R\'{e}nyi entropy when the linearly normalized partial maxima converges to some nondegenerate random variable.

\vspace{0.5in}

\vspace{0.2in} \noindent {\bf Keywords:} R\'{e}nyi entropy, Max stable laws, Max domain of attraction.

\vspace{0.5in}

\vspace{0.2in} \noindent {\bf MSC 2010 classification:} 60F10
\newpage
\section{Introduction}
The limit laws of linearly normalized partial maxima $M_n=\max(X_1,\cdots, X_n)$ of independent and identically distributed (iid) random variables (rvs) $X_1,X_2,\ldots,$ with common distribution function (df) $F,$ namely,
\begin{equation}\label{Introduction_e1}
	\lim_{n\to\infty}\Pr(M_n\leq a_nx+b_n)=\lim_{n\to\infty}F^n(a_nx+b_n)=G(x),\;\;x\in \mathcal{C}(G),
\end{equation}
	where, $a_n>0,$ $b_n\in\Real,$ are norming constants, $G$ is a non-degenerate distribution function, $\mathcal{C}(G)$ is the set of all continuity points of $G,$ are called max stable laws. If, for some non-degenerate distribution function $G,$ a distribution function $F$ satisfies (\ref{Introduction_e1}) for some norming constants $a_n>0,$ $b_n\in\Real,$ then we say that $F$ belongs to the max domain of attraction of $G$ under linear normalization and denote it by $F\in \mathcal{D}(G).$ Limit distribution functions $G$ satisfying (\ref{Introduction_e1}) are the well known extreme value types of distributions, or max stable laws, namely,
	\begin{eqnarray*}
			\text{the Fr\'{e}chet law:} & \Phi_\alpha(x)  = \left\lbrace	
							\begin{array}{l l}
							 0, &\;\;\; x< 0, \\
							 \exp(-x^{-\alpha}), &\;\;\; 0\leq x;\\
							 \end{array}
							 \right. \\					
		\text{the Weibull law:} & \Psi_\alpha(x) =  \left\lbrace
						\begin{array}{l l} \exp(- |x|^{\alpha}), & x<0, \\
						1, & 0\leq x;
						\end{array}\right. \\
		\text{and the Gumbel law:} & \Lambda(x) = \exp(-\exp(-x));\;\;\;\;\; x\in\Real;
	\end{eqnarray*}
$\alpha>0$ being a parameter, with respective density functions,
	\begin{eqnarray*}
		\text{the Fr\'{e}chet density:} & \phi_\alpha(x) =  \left\lbrace	
							\begin{array}{l l}
							 0, &\;\;\; x \leq 0, \\
							 \alpha x^{-(\alpha+1)}e^{-x^{-\alpha}}, &\;\;\; 0 < x;\\
							 \end{array}
							 \right. \\
		\text{the Weibull density:} & \psi_\alpha(x) = \left\lbrace
						\begin{array}{l l} \alpha |x|^{\alpha-1}e^{-|x|^{\alpha}}, & x<0, \\
						0, & 0 \leq x;
						\end{array}\right. \\
		\text{and the Gumbel density:} & \lambda(x) = e^{-x}e^{-e^{-x}},\;\; x\in\Real.
	\end{eqnarray*}
Note that (\ref{Introduction_e1}) is equivalent to
\begin{eqnarray}
\lim_{n\to\infty}n(1-F(a_nx+b_n))=-\log G(x), \; x \in \{y: G(y) > 0\}.\nonumber
\end{eqnarray}
We shall denote the left extremity of distribution function $F$ by $l(F) = \inf\{x: F(x) > 0\} \geq - \infty $ and the right extremity of $F$ by $r(F) = \sup \{x: F(x) < 1\} \leq \infty.$
Criteria for $F\in \mathcal{D}(G)$ are well known (see, for example, Galambos, 1987; Resnick, 1987; Embrechts et al., 1997).

The Shannon entropy of a continuous rv $X$ with density function $f(x)$ is defined as 
\begin{eqnarray*}
	H(X) =   -\int_A f(x)\log f(x)dx, \;\;\mbox{where}\;\;A=\{x\in\Real:f(x)>0\}.
\end{eqnarray*}
 R\'{e}nyi entropy is a generalization of Shannon entropy (R\'{e}nyi, 1961). It is one of a family of functional for quantifying the diversity, uncertainty or randomness of a system. The R\'{e}nyi entropy of order $\beta $ is defined as
\begin{equation}\label{Renyi_entropy}
H_\beta(X) = \frac{1}{1-\beta}\log\left(\int_A (f(x))^\beta\,dx\right) ,
\end{equation}
where, $0<\beta<\infty,$ $\beta\neq 1.$ By L'Hopital's rule, the R\'{e}nyi entropy tends to Shannon entropy, as $\beta\to 1.$ The R\'{e}nyi entropy are important in ecology and statistics as indices of diversity. R\'{e}nyi entropy appear also in several important contexts such as those of information theory, statistical estimation.

 The idea of tracking the central limit theorem using Shannon entropy goes back to Linnik (1959) and Shimizu (1975), who used it to give a particular proof of the central limit theorem. Brown (1982), Barron (1986) and Takano (1987) discuss the central limit theorem with convergence in the sense of Shannon entropy and relative entropy. Artstein et al.(2004) and Johnson and Barron (2004) obtained the rate of convergence under some conditions on the density. Johnson (2006) is a good reference to the application of information theory to limit theorems, especially the central limit theorem. Cui and Ding (2010) show  that the convergence of the R\'{e}nyi entropy of the normalized sums of iid rvs and obtain the corresponding rates of convergence.
Saeb (2014) study the rate of convergence of R\'{e}nyi entropy for the max domain of attraction.

In this article, our main interest is to investigate conditions under which the R\'{e}nyi entropy of the normalized partial maxima of iid rvs converges to the corresponding limit R\'{e}nyi entropy. In the other hand, our problem of interest to see if normalized partial maxima converges to a nondegenerate rv, does the R\'{e}nyi entropy of the normalized partial maxima converges to the R\'{e}nyi entropy of the limit rv?
In the next section we give our main results, followed by a section on Proofs. The R\'{e}nyi entropies of the extreme value distributions are given in the appendix A with proof and appendix B  containing results used in this article.

\section{Main Result}
\begin{thm}\label{Theorem}
 Suppose $F\in D(G)$ is absolutely continuous with pdf $f$ which is eventually positive and decreasing in left neighbourhood of $r(F).$ If  $\int_{-\infty}^{r(F)}(f(x))^{\beta}dx<\infty$ for $\beta>1,$ and\\
\item[(a)] $G=\Phi_\alpha$  and  $r(F)=\infty,$ then 
$\lim_{n \rightarrow \infty} H_\beta(g_n) = H_\beta(\phi_\alpha);$\\
\item[(b)] $G=\Psi_\alpha$ and $r(F)<\infty,$ then $\lim_{n \rightarrow \infty} H_\beta(g_n) = H_\beta(\psi_\alpha);$\\
\item[(c)] $G=\Lambda$ and $r(F)\leq\infty,$ then
$\lim_{n \rightarrow \infty} H_\beta(g_n) = H_\beta(\lambda).$
\end{thm}
\begin{rem}
Lemma \ref{R_ent_general1} show that, the R\'{e}nyi entropy  do not depend on the location and scale parameters.
\end{rem}
The proof of the above theorem is different from the proof for the R\'{e}nyi entropy of the normalized sums of iid rvs. In our proofs, the properties of normalized partial maxima such as, von Mises condition and density convergence, plays an important role.
\section{Proofs}
\begin{lem}\label{lem1}
If $\bar{F}\in RV_{-\alpha}$ then $a(\cdot)\in RV_{1/\alpha}.$
\end{lem}

\begin{proof}
We know that $1 - F$ is regularly varying so that
$n\bar{F}(a_nx)\to x^{-\alpha}$ as $n\to\infty.$
Set $U=1/\bar{F}$ and $V=U^{\leftarrow}.$ Therefore
\begin{eqnarray}
U(a_nx)/n\to x^\alpha,\;\;\text{for }x>0,\nonumber
\end{eqnarray}
and inverting we have
\begin{eqnarray}
V(ny)/a_n\to y^{-1/\alpha}.\nonumber
\end{eqnarray}
Since, $a_n\simeq (\frac{1}{1-F})^{\leftarrow}(n)=V(n)$ and a switch to a continuous variable we have
\begin{eqnarray}
\frac{V(ty)-V(t)}{a(t)}\to y^{1/\alpha}-1.\nonumber
\end{eqnarray}
Now we have,
\begin{eqnarray}
\lim_{t\to\infty}\frac{a(tx)}{a(t)}&=&\lim_{t\to\infty}\frac{a(tx)(V(tx)-V(t))}{-a(t)(V(tx\,x^{-1}-V(tx)))},\nonumber\\
&=&(x^{1/\alpha}-1)/(1-x^{-1/\alpha}),\nonumber\\
&=&x^{1/\alpha}.\nonumber
\end{eqnarray}
\end{proof}

\begin{proof}[\textbf{Proof of Theorem \ref{Theorem}-(a).}]
Suppose $F \in \mathcal{D}(\Phi_\alpha),$  and  $1 - F$ is regularly varying so that
\begin{eqnarray}
\lim_{n \rightarrow \infty} \frac{\overline{F}(a_nx)}{\overline{F}(a_n)} & = &  x^{- \alpha}, \;\; x > 0; \;\;\nonumber \mbox{and}\\
\lim_{n \rightarrow \infty} F^n(a_n x) & = & \Phi_{\alpha}(x), \; x \in \Real, \nonumber
\end{eqnarray}
with $\; a_n = F^{\leftarrow}(1 - \frac{1}{n}) = \inf\{x: F(x) > 1 - \frac{1}{n}\}, \, n \geq 1\;$ and $b_n=0.$
From Theorem \ref{thm_von}, $F$ satisfies the von Mises condition:
\begin{equation} \label{vonMises}
\lim_{t \rightarrow \infty} \frac{tf(t)}{1-F(t)} = \alpha.
\end{equation}
Now, by Theorem \ref{gn_conv}, implies the following density convergence on compact sets:
\begin{equation} \label{s_DenConv_f}
\lim_{n \rightarrow \infty} g_n(x) = \phi_\alpha(x),\;\;x\in K\subset (0,\infty)
\end{equation} where $K$ is a compact set, and $\;g_n(x)=na_nf(a_nx)F^{n-1}(a_nx).\;$

From definition of R\'{e}nyi entropy, we write,
\begin{eqnarray}\label{F_Hgn}
H_\beta(g_n)&=&\frac{1}{1-\beta}\log\left[I_A(n,v)+I_B(n,v)+I_C(n,v)\right].
\end{eqnarray}
where, $I_A(n,v)=\int_{v}^{\infty}(g_n(x))^\beta\,dx, $ and  $I_B(n,v)=\int_{-\infty}^{v^{-1}}(g_n(x))^{\beta}\,dx,$ and $I_C(n,v)=\int_{v^{-1}}^{v}(g_n(x))^\beta\,dx.$ It is enough to show that
\begin{eqnarray}
\lim_{v\to\infty}\lim_{n\to\infty}\left(I_A(n,v)+I_B(n,v)\right)=0.\nonumber
\end{eqnarray}
We set,
\begin{eqnarray}
0<I_A(n,v)&=& \int_{v}^{\infty}(g_n(x))^{\beta-1}\,dF^n(a_nx),\nonumber\\
&=&L(n,\beta)\int_{v}^{\infty}\left(\frac{f(a_nx)}{f(a_n)}F^{(n-1)}(a_nx)\right)^{\beta-1}\,g_n(x)dx,\nonumber\\
&<&L(n,\beta)\left(\frac{f(a_nv)}{f(a_n)}\right)^{\beta-1}(1-F^{n}(a_nv)),\nonumber
\end{eqnarray}
where, $L(n,\beta)=\left(\dfrac{a_nf(a_n)n\bar{F}(a_n)}{\bar{F}(a_n)}\right)^{\beta-1},$ and, $n\bar{F}(a_n)=1$ and using von Mises conditions in (\ref{vonMises}) $f$ is decreasing function, $L(n,\beta)\to \alpha^{\beta-1},$ as $n\to\infty,$ and $\frac{f(a_nx)}{f(a_n)}<1$ for $x\geq 1.$
Hence, for $\beta>1,$
\begin{eqnarray}
\lim_{v\to \infty}\lim_{n\to\infty}I_A(n,v)=0.\label{F_IA}
\end{eqnarray}
Now, we choose $\xi_n$ by $-\log F(\xi_n)\simeq n^{-1/2},$ and $t_n=\frac{\xi_n}{a_n}.$ If $\frac{\xi_n}{a_n}\to c>0$ then $n^{1/2}\simeq -n(\log F(t_n\,a_n))\to c^{-\alpha}$ and this is contradict the fact that $n^{1/2}\to\infty.$ Therefore, $t_n\to 0$ as $\xi_n\to\infty$ for large $n.$\\
We have,
\begin{eqnarray}I_B(n,u)=\int_{-\infty}^{t_n}(g_n(x))^{\beta}dx+\int_{t_n}^{v^{-1}}(g_n(x))^{\beta}dx=I_{B_1}(n)+I_{B_2}(n,v).\label{F_IB}
\end{eqnarray}
 We set,
\begin{eqnarray}
I_{B_1}(n)&=&n^{\beta}a_n^{\beta-1}\int_{-\infty}^{\xi_n}(F^{n-1}(s)f(s))^{\beta}\,ds,\;\;(\text{where, } a_nx=s),\nonumber\\
&\leq&n^{\beta}a_n^{\beta-1}F^{\beta(n-1)}(\xi_n)\int_{-\infty}^{\infty}(f(s))^{\beta}\,ds,\nonumber\\
&\simeq&\frac{n^{\beta}a_n^{\beta-1}}{\exp\{(n-1)\beta  n^{-1/2}\}}\int_{-\infty}^{\infty}(f(s))^{\beta}\,ds,\nonumber
\end{eqnarray}
Since $a_n\simeq \left(\frac{1}{1-F}\right)^{\leftarrow}(n)$ from Lemma \ref{lem1}, $a_n\in RV_{\frac{1}{\alpha}}$ and (\ref{kara}) for $n>N$ given $\epsilon>0$ and $\rho(n)<\frac{1+\epsilon}{\alpha}$ then
\begin{eqnarray}
a_n&=&c(n)\exp\Big\{\int_{N}^{n}\rho(t)
t^{-1}dt\Big\},\nonumber\\
&<&(n/N)^{\frac{1+\epsilon}{\alpha}}c,\nonumber
\end{eqnarray}
 Therefore,
\begin{eqnarray}
0< I_{B_1}(n)&\leq&\frac{c\,n^{\beta} (n/N)^{\frac{1+\epsilon}{\alpha}(\beta-1)}}{\exp\{\beta( n^{1/2}-n^{-1/2}\}}\int_{-\infty}^{\infty}(f(s))^{\beta}\,ds.\nonumber
\end{eqnarray}  
If $\int_{-\infty}^{\infty}(f(x))^\beta dx<\infty$ then,
\begin{eqnarray}
\lim_{n\to\infty}I_{B_1}(n)=0.\label{F_IB1}
\end{eqnarray}
From (\ref{von_F}), for given $\epsilon_1>0$ we have $f(a_nx)\leq -(\alpha+\epsilon_1)\frac{\log F(a_nx)}{a_nx}$ ultimately and from Theorem \ref{rv} for sufficiently large $n$ given $\epsilon_2>0$ and  $n\bar{F}(a_nx)<\frac{1}{1-\epsilon_2}x^{-(\alpha-\epsilon_2)}$ and given $\epsilon_3>0$ and for all $n>1$ then $-\frac{n-1}{n}<-(1-\epsilon_3).$ Hence, for sufficiently large $n$ such that $n>n_0$ we have
\begin{eqnarray}
g_n(x)&=&na_nf(a_nx)F^{n-1}(a_nx),\nonumber\\
&<&-(\alpha+\epsilon_1)n\log F(a_nx)x^{-1}\exp\Big\{\frac{n-1}{n}n\log F(a_nx)\Big\},\nonumber\\
&<&\frac{\alpha+\epsilon_1}{1-\epsilon_2}x^{-1-\alpha+\epsilon_2}\exp\Big\{-\frac{1-\epsilon_3}{1-\epsilon_2}x^{-(\alpha-\epsilon_2)}\Big\},\nonumber\\
&<&c\alpha' x^{-\alpha'-1}\exp\Big\{-cx^{-\alpha'}\Big\}.\nonumber
\end{eqnarray}
where, $\alpha'=\alpha-\epsilon_2,$ and $c$ is positive constant. We define 
\begin{eqnarray}
h(x)=c\alpha' x^{-\alpha'-1}\exp\Big\{-cx^{-\alpha'}\Big\}.\label{F_h}
\end{eqnarray}
 Set, $I_{B_2}(n,v)=\int_{t_n}^{v^{-1}}(g_n(x))^{\beta}\,dx.$ From (\ref{F_h}) for large $n,$ $g_n(x)<h(x)$ we have
\begin{eqnarray}
0<I_{B_2}(n,v)&<&\int_{0}^{v^{-1}}(h(x))^\beta dx,\nonumber
\end{eqnarray}
Since $\int_{0}^{\infty}(h(x))^\beta<\infty$ so that,
\begin{eqnarray}
\lim_{v\to\infty}\lim_{n\to\infty} I_{B_2}(n,v)=0.\label{F_IB2}
\end{eqnarray}
From, (\ref{F_IB}), (\ref{F_IB1}) and (\ref{F_IB2})
\begin{eqnarray}
\lim_{v\to\infty}\lim_{n\to\infty}I_{B}(n,v)=0.\label{F_IB3}
\end{eqnarray}
Next, $I_C(n,v)=\int_{v^{-1}}^{v}(g_n(x))^\beta\,dx,$ and from (\ref{F_h}) for large $n,$ we have $g_n(x)<h(x),$  and $\int_{v^{-1}}^{v}h(x)dx<\infty,$ and using (\ref{s_DenConv_f}) $\lim_{n\to\infty}g_n(x)=\phi_\alpha(x),$ locally uniformly convergence in $x\in [v^{-1},v]$ by DCT,
\begin{eqnarray}\label{F_IC}
\lim_{v\to\infty}\lim_{n\to\infty}\int_{v^{-1}}^{v}(g_n(x))^\beta\,dx&=&\int_{0}^{\infty}(\phi_\alpha(x))^\beta\,dx.
\end{eqnarray}
And from (\ref{F_Hgn}), (\ref{F_IA}), (\ref{F_IB3}) and (\ref{F_IC}) imply,
\begin{eqnarray}
\lim_{n\to\infty}H_\beta(g_n)&=&  H_\beta(\phi_{\alpha}).\nonumber
\end{eqnarray}
\end{proof}

Let $Y_1, Y_2,\ldots$ are iid rvs with common distribution function $F_Y$ and $r(F_Y)<\infty$ and $X_i=1/(r(F_Y)-Y_i)$ with common distribution function $F_X.$ In following lemma is easily show that the relationship between the domain attraction of $\Phi_\alpha$ and $\Psi_\alpha.$

\begin{lem}\label{Lemma3_w}
If $F_Y\in\D(\Psi_\alpha)$ then $F_X\in \D(\Phi_\alpha)$ with $a_n=\frac{1}{\delta_n}$ and $b_n=0.$
\end{lem}

\begin{proof}
From $F_Y\in\d(\Psi_\alpha)$ with $\delta_n>0$ and $r(F_Y)<\infty,$ we have, $\dfrac{\vee_{i=1}^{n}Y_i-r(F_Y)}{\delta_n}\xrightarrow{d} M,$ where, $M$ is rv with distribution function $\Psi_\alpha.$ Therefore,
\begin{eqnarray}
-\left(\frac{\vee_{i=1}^{n}Y_i-r(F_Y)}{\delta_n}\right)^{-1}
&=&\left(\vee_{i=1}^{n}\frac{\delta_n}{r(F_Y)-Y_i}\right),\nonumber\\
&=&\vee_{i=1}^{n}\frac{X_i}{a_n}\xrightarrow{d} -\frac{1}{M},\label{e_Lemma3_w}
\end{eqnarray}
and $-\frac{1}{M}$ is a rv with distribution function $\Phi_\alpha$ and left of (\ref{e_Lemma3_w}) is equivalent to $F_X\in\D(\Phi_\alpha)$ with $a_n=\frac{1}{\delta_n}$ and $b_n=0.$ 
\end{proof}
\begin{proof} [\textbf{Proof of Theorem \ref{Theorem}-(b).}]
Suppose $F \in \mathcal{D}(\Psi_\alpha),$ iff $r(F)<\infty$ and $\overline{F}(r(F)-x^{-1})\in RV_{-\alpha},$ is regularly varying. In this case we may set $\tau_n=F^{-1}(1 - \frac{1}{n}),$ and $\delta_n=(r(F)-\tau_n),$ and then
\begin{eqnarray}
\lim_{n \rightarrow \infty} F^n(\delta_n x+r(F)) = \Psi_{\alpha}(x), \; x \in\Real,\nonumber
\end{eqnarray}
From Theorem \ref{von_W}, $F$ satisfies the von Mises condition:
\begin{eqnarray} 
\lim_{x \uparrow r(F)} \frac{(r(F)-x)f(x)}{\overline{F}(x)} = \alpha.\nonumber
\end{eqnarray}
Now, by Theorem \ref{gn_conv}, implies the following density convergence on compact sets:
\begin{equation}
\lim_{n \rightarrow \infty} g_n(x) = \psi_{\alpha}(x), \; x \in K \subset (-\infty,0),\nonumber
\end{equation} where $\; K \;$ is a compact set, and $g_n(x)=n\delta_n f(\delta_n x+r(F))F^{n-1}(\delta_n x+r(F)).$
From definition of R\'{e}nyi entropy, we write,
\begin{eqnarray}
H_\beta(g_n)&=&\frac{1}{1-\beta}\log\left(\int_{-\infty}^{0}(g_n(x))^\beta\,dx\right).\nonumber
\end{eqnarray}
Put $x=-1/y,$ and using Lemma \ref{Lemma3_w}, we have,
\begin{eqnarray}
H_\beta(g_n)&=&\frac{1}{1-\beta}\log\left(\int_{0}^{\infty}
(\tilde{g}_n(x))^{\beta}\frac{dy}{y^2}\right),\nonumber
\end{eqnarray}
where, $\tilde{g}_n(x)=n\delta_nf(r(F)-\delta_n/y)F^{n-1}(r(F)-\delta_n/y).$ From Lemma \ref{Lemma3_w}, $F_X(a_ny)=F(r(F)-\delta_n/y)\in\D(\Phi_\alpha),$ and with conditions Theorem \ref{Theorem}-(a), if $\int_{-\infty}^{r(F)}(f(s))^{\beta}ds<\infty$ and $f$ is decreasing function for $\beta>1$ then
\begin{eqnarray}
\lim_{n\to\infty}H_\beta(g_n)&=& H_\beta(\psi_{\alpha}).\nonumber
\end{eqnarray}
\end{proof}

\begin{lem}\label{G_Lemma6}
Suppose $F\in\D(\Lambda)$ with auxiliary function $u$ and $\epsilon>0.$ There exists a large $N$ such that for $x\geq 0,$ and $n>N$
\begin{eqnarray}
n(1-F(a_nx+b_n))\leq(1+\epsilon)^2(1-\epsilon x)^{\epsilon^{-1}},\nonumber
\end{eqnarray}
and for $x<0,$
\begin{eqnarray}
n(1-F(a_nx+b_n))\leq(1+\epsilon)^2(1+\epsilon \abs{x})^{\epsilon^{-1}}.\nonumber
\end{eqnarray}
\end{lem}

\begin{proof}
From Theorem \ref{Lem5_appendix} and for $x\geq 0,$ and sufficient large $n$ such that $\abs{u'(a_nt+b_n)}\leq \epsilon$
\begin{eqnarray}
\abs{\frac{u(a_nx+b_n)}{u(b_n)}-1}&=&\abs{\int_{b_n}^{a_nx+b_n}\frac{u'(s)}{u(b_n)}ds},\nonumber\\
&\leq &\int_{0}^{x}\abs{u'(a_nt+b_n)}dt,\;\;\text{where }a_ns+b_n=t,\nonumber\\
&\leq&\epsilon\,x.\label{ratio}
\end{eqnarray}
Consequently, $\frac{u(b_n)}{u(a_nx+b_n)}<\frac{1}{1-\epsilon\,x}.$ Since, for large $n,$ and $\epsilon>0$ we have $\frac{(1-\epsilon)}{\bar{F}(b_n)}<n<\frac{(1+\epsilon)}{\bar{F}(b_n)}$ then
\begin{eqnarray}
n(1-F(a_nx+b_n))&\leq&(1+\epsilon)\frac{1-F(b_n+a_nx)}{1-F(b_n)},\nonumber\\
&=&(1+\epsilon)\frac{c(b_n+a_nx)}{c(b_n)}\exp\Big\{-\int_{b_n}^{b_n+a_nx}\frac{dy}{u(y)}\Big\},\nonumber\\
&\leq&(1+\epsilon)^2\exp\Big\{-\int_{0}^{x}\frac{u(t)ds}{u(b_n+a_ns)}\Big\},\nonumber
\end{eqnarray}
where, $y=b_n+a_ns$ and $\lim_{t\to \infty}\frac{c(t+xu(t))}{c(t)}=1$ and from (\ref{ratio}), then
\begin{eqnarray}
n(1-F(a_nx+b_n))&\leq&(1+\epsilon)^2\exp\Big\{-\int_{0}^{x}\frac{dt}{1-\epsilon t}\Big\},\nonumber\\
&=&(1+\epsilon)^2(1-\epsilon x)^{\epsilon^{-1}}.\nonumber
\end{eqnarray}

For the second statement, for $x<0$ and sufficient large $n$ 
\begin{eqnarray}
\abs{1-\frac{u(a_nx+b_n)}{u(b_n)}}&=&\abs{\int_{a_nx+b_n}^{b_n}\frac{u'(s)}{u(b_n)}ds},\nonumber\\
&\leq &\int_{x}^{0}\abs{u'(a_nt+b_n)}dt,\;\;\text{where }a_ns+b_n=t,\nonumber\\
&\leq&\epsilon\,\abs{x}.\label{ratio2}
\end{eqnarray}
Consequently, $\frac{u(b_n)}{u(a_nx+b_n)}< \frac{1}{1+\epsilon\,\abs{x}}.$ Then,
\begin{eqnarray}
n(1-F(a_nx+b_n))&\leq&(1+\epsilon)\frac{c(b_n+a_nx)}{c(b_n)}\exp\Big\{-\int_{b_n}^{b_n+a_nx}\frac{dy}{u(y)}\Big\},\nonumber\\
&\leq&(1+\epsilon)^2\exp\Big\{\int_{x}^{0}\frac{u(t)ds}{u(b_n+a_ns)}\Big\},\nonumber
\end{eqnarray}
where, $y=b_n+a_ns$ and $\lim_{t\to \infty}\frac{c(t+xu(t))}{c(t)}=1$ and from (\ref{ratio2}), then
\begin{eqnarray}
n(1-F(a_nx+b_n))&\leq&(1+\epsilon)^2\exp\Big\{\int_{x}^{0}\frac{dt}{1+\epsilon \abs{t}}\Big\},\nonumber\\
&=&(1+\epsilon)^2(1+\epsilon \abs{x})^{\epsilon^{-1}}.\nonumber
\end{eqnarray}
\end{proof}
\begin{proof}[\textbf{Proof of Theorem \ref{Theorem}-(c).}]

Suppose $F \in \mathcal{D}(\Lambda)$ and $1 - F$ is $\Gamma$ varying so that
\begin{eqnarray} \label{s_RegVar_g}
\lim_{n \rightarrow\infty} \frac{\overline{F}(b_n+xa_n)}{\overline{F}(b_n)} & = &  e^{-x}, \;\; x\in\Real,\nonumber
\end{eqnarray}
where the function $a_n=u(b_n)=\int_{b_n}^{r(F)}\overline{F}(s)ds/\overline{F}(b_n)$ is called an auxiliary function. From (\ref{Introduction_e1})
\begin{eqnarray}
\lim_{n \rightarrow \infty} F^n(a_n x+b_n) & = & \Lambda(x), \; x \in\Real,  \nonumber
\end{eqnarray}
with $\;a_n = u(b_n)$ and $b_n=F^{\leftarrow}(1 - \frac{1}{n}) = \inf\{x: F(x) > 1 - \frac{1}{n}\}, \, n \geq 1.\;$
From Theorem \ref{thm_g_von}, $F$ satisfies the von Mises condition:
\begin{equation} \label{s_vonMises_g}
\lim_{n \rightarrow \infty} \dfrac{f(b_n)u(b_n)}{\overline{F}(b_n)} = 1.
\end{equation}
Now, by Theorem \ref{gn_conv} implies the following density convergence on compact sets:
\begin{equation} \label{s_DenConv_g}
\lim_{n \rightarrow \infty} g_n(x) = \lambda(x), \; x \in K \subset \Real,
\end{equation} where $\; K \;$ is a compact set, and $\;g_n(x)=n\,a_n\,f(a_nx+b_n)F^{n-1}(a_nx+b_n).\;$

We write
\begin{eqnarray}
H_\beta(g_n)&=\frac{1}{(1-\beta)}\log\left[\int_{-v}^{v}(g_n(x))^\beta dx+\int_{-\infty}^{-v}(g_n(x))^\beta dx+\int_{v}^{\infty}(g_n(x))^\beta dx\right].\nonumber
\end{eqnarray}
It is enough to show,
\begin{eqnarray}
\lim_{v\to\infty}\lim_{n\to\infty}\left[\int_{-\infty}^{-v}(g_n(x))^\beta dx+\int_{v}^{\infty}(g_n(x))^\beta dx\right]=0.\nonumber
\end{eqnarray}
Set, $I_E(n,v)=\int_{v}^{\infty}(g_n(x))^{\beta-1} g_n(x)dx,$ we have
\begin{eqnarray}
0<I_E(n,v)<J(n,\beta)\left(\frac{f(a_nv+b_n)}{f(b_n)}\right)^{\beta-1}(1-F^n(a_nv+b_n)).\nonumber
\end{eqnarray}
where, $J(n,\beta)= \left(n\bar{F}(b_n)\frac{f(b_n)a_n}{\bar{F}(b_n)} \right)^{\beta-1}\to 1$ as $n\to\infty$ and, $0\leq F(.)\leq 1,$ for $\beta>1$ and $f$ is decreasing function then $\frac{f(a_nv+b_n)}{f(b_n)}<1$ for $v\geq 1.$
Therefore,
\begin{eqnarray}
\lim_{v\to\infty}\lim_{n\to\infty}I_E(n,v)=0.\label{G_IE}
\end{eqnarray}
Now, we choose $\xi_n$ satisfying $-\log F(\xi_n)\simeq n^{-1/2}.$ If $t_n=\frac{\xi_n-b_n}{a_n}\to c,$ then $n^{1/2}\simeq-n\log F(a_nt_n+b_n)\to e^{-c},$ and this is contradict with $n^{1/2}\to\infty.$ Therefore, $t_n=\frac{\xi_n-b_n}{a_n}\to -\infty,$ as $\xi_n\to r(F)$  for large $n.$\\
We now decompose the integral, 
$$I_F(n,v)=\int_{-\infty}^{t_n}(g_n(x))^\beta dx+\int_{t_n}^{-v}(g_n(x))^\beta dx=I_{F_1}(n)+I_{F_2}(n,v).$$
Set
\begin{eqnarray}
0<I_{F_1}(n)&<&n^{\beta}a_n^{\beta-1} F^{\beta(n-1)}(\xi_n)\int_{\Real}(f(x))^\beta dx,\nonumber\\
&\simeq&\frac{n^{\beta}a_n^{\beta-1}}{\exp\{n^{1/2}\beta\}}\exp\{\beta\,n^{-1/2}\}\int_{-\infty}^{\infty}(f(x))^\beta dx,\nonumber
\end{eqnarray}
where, $F^{n-1}(\xi_n)=\exp\Big\{(n-1)\log F(\xi_n)\Big\}\simeq\exp\{-n^{1/2}+n^{-1/2}\}.$ From Theorem \ref{a_rv}, $a_n\in RV_0$ and from Lemma \ref{slow}, for $n>N$ given $\epsilon>0$ and $\epsilon(n)<\epsilon$ then
\begin{eqnarray}
a_n=c(n)\exp\Big\{\int_{N}^{n}\rho(t)t^{-1}dt\Big\}< (n/N)^{\epsilon}c,\nonumber
\end{eqnarray}
 Therefore,
\begin{eqnarray}
0< I_{F_1}(n)&\leq&\frac{c n^{\beta} (n/N)^{(\beta-1)\epsilon}}{\exp\{\beta(n^{1/2}-n^{-1/2})\}}\int_{-\infty}^{\infty}(f(s))^{\beta}\,ds.\nonumber
\end{eqnarray}  
If $\int_{-\infty}^{\infty}(f(x))^\beta dx<\infty$ then
\begin{eqnarray}
\lim_{n\to\infty}I_{F_1}(n)=0.\label{G_IF1}
\end{eqnarray}
From (\ref{s_vonMises_g}), for given $\epsilon_1>0$ we have $f(a_nx+b_n)\leq -(1+\epsilon_1)\frac{\log F(a_nx+b_n)}{u(a_nx+b_n)}$ ultimately and from Theorem \ref{Lem4_appendix} given $\epsilon_2>0$ and for $x<0$ such that $a_nx+b_n\geq n_0,$ $\frac{u(b_n)}{u(a_nx+b_n)}<\frac{1}{1-\epsilon_2}\left[\frac{-\log F(a_nx+b_n)}{-\log F(b_n)}\right]^{\epsilon_2}$ and for $\epsilon_3>0$ and for all $n\geq 1$ then $1-\epsilon_3<\frac{n-1}{n},$ and from Lemma \ref{G_Lemma6} for $\epsilon_4>0$ and $x<0$ we have $n\bar{F}(a_nx+b_n)<(1+\epsilon_4)^2(1+\epsilon_4\abs{x})^{\epsilon_4^{-1}},$ and for $\epsilon_5>0,$ $\frac{1}{\bar{F}(b_n)}<\frac{n}{1-\epsilon_5}.$ Therefore,
\begin{eqnarray}
g_n(x)&=&na_nf(a_nx+b_n)F^{n-1}(a_nx+b_n),\nonumber\\
&<&-(1+\epsilon_1)\frac{u(b_n)}{u(a_nx+b_n)}n\log F(a_nx+b_n)\exp\Big\{\frac{n-1}{n}n\log F(a_nx+b_n)\Big\},\nonumber\\
&<&\frac{1+\epsilon_1}{(1-\epsilon_2)(1-\epsilon_5)^{1+\epsilon_2}}(-n\log F(a_nx+b_n))^{1+\epsilon_2}\exp\Big\{(1-\epsilon_3)n\log F(a_nx+b_n)\Big\},\nonumber\\
&<&\frac{(1+\epsilon_1)(1+\epsilon_4)^{2\frac{1+\epsilon_2}{\epsilon_4}}}{(1-\epsilon_2)(1-\epsilon_5)^{1+\epsilon_2}}(1+\epsilon_4\abs{x})^{\frac{1+\epsilon_2}{\epsilon
_4}}\exp\Big\{-(1-\epsilon_3)(1+\epsilon_4)^2(1+\epsilon_4\abs{x})^{\epsilon_4^{-1}}\Big\},\nonumber\\
&=&c_1(1+\epsilon_4\abs{x})^{c_2}
\exp\Big\{-c_3 (1+\epsilon_4\abs{x})^{\epsilon_4^{-1}}\Big\},\nonumber
\end{eqnarray}
where, $c_1=\frac{(1+\epsilon_1)(1+\epsilon_4)^{2\frac{1+\epsilon_2}{\epsilon_4}}}{(1-\epsilon_2)(1-\epsilon_5)^{1+\epsilon_2}}$ and $c_2=\frac{1+\epsilon_2}{\epsilon
_4},$ and $c_3=(1-\epsilon_3)(1+\epsilon_4)^2.$
We define 
\begin{eqnarray}\label{G_h}
h_1(x)=c_1(1+\epsilon_4\abs{x})^{c_2}
\exp\Big\{-c_3 (1+\epsilon_4\abs{x})^{\epsilon_4^{-1}}\Big\}.
\end{eqnarray}
And
\begin{eqnarray}
\int_{-\infty}^{0}h_1(x)dx&=&\int_{-\infty}^{0}
c_1(1+\epsilon_4\abs{x})^{c_2}
\exp\Big\{-c_3 (1+\epsilon_4\abs{x})^{\epsilon_4^{-1}}\Big\}dx.\nonumber
\end{eqnarray}
Putting, $c_3(1+\epsilon_4 \abs{x})^{\epsilon_4^{-1}}=y,$ and $-c_3(1+\epsilon_4 \abs{x})^{\epsilon_4^{-1}-1}dx=dy$ we have,
\begin{eqnarray}
\int_{-\infty}^{0}h_1(x)dx&=&c_1\left(\frac{1}{c_3}\right)^{\epsilon_4( c_2+1)}\int_{c_3}^{\infty}y^{\epsilon_4\,(c_2+1)-1}\,e^{-y}dy<\infty.\nonumber
\end{eqnarray}
Similarly, from (\ref{s_vonMises_g}) and from Theorem \ref{Lem4_appendix} given $\epsilon_2>0$ and for $x>0$ such that $a_nx+b_n\geq n_0,$ $\frac{u(b_n)}{u(a_nx+b_n)}<\frac{1}{1-\epsilon_2}\left[\frac{-\log F(a_nx+b_n)}{-\log F(b_n)}\right]^{-\epsilon_2}$ and from Lemma \ref{G_Lemma6} for $\epsilon_4>0$ and $x>0$ we have $n\bar{F}(a_nx+b_n)<(1+\epsilon_4)^2(1-\epsilon_4x)^{\epsilon_4^{-1}},$ we have
\begin{eqnarray}
g_n(x)&=&na_nf(a_nx+b_n)F^{n-1}(a_nx+b_n),\nonumber\\
&<&-(1+\epsilon_1)\frac{u(b_n)}{u(a_nx+b_n)}n\log F(a_nx+b_n)\exp\Big\{\frac{n-1}{n}n\log F(a_nx+b_n)\Big\},\nonumber\\
&<&\frac{1+\epsilon_1}{(1-\epsilon_2)(1-\epsilon_5)^{1-\epsilon_2}}(-n\log F(a_nx+b_n))^{1-\epsilon_2}\exp\Big\{(1-\epsilon_3)n\log F(a_nx+b_n)\Big\},\nonumber\\
&<&\frac{(1+\epsilon_1)(1+\epsilon_4)^{2\frac{1-\epsilon_2}{\epsilon_4}}}{(1-\epsilon_2)(1-\epsilon_5)^{1-\epsilon_2}}(1-\epsilon_4 x)^{\frac{1-\epsilon_2}{\epsilon
_4}}\exp\Big\{-(1-\epsilon_3)(1+\epsilon_4)^2(1-\epsilon_4 x)^{\epsilon_4^{-1}}\Big\},\nonumber\\
&=&c_1(1-\epsilon_4x)^{c_2}
\exp\Big\{-c_3 (1-\epsilon_4 x)^{\epsilon_4^{-1}}\Big\},\nonumber
\end{eqnarray}
where, $c_1=\frac{(1+\epsilon_1)(1+\epsilon_4)^{2\frac{1-\epsilon_2}{\epsilon_4}}}{(1-\epsilon_2)(1-\epsilon_5)^{1-\epsilon_2}}$ and $c_2=\frac{1-\epsilon_2}{\epsilon
_4},$ and $c_3=(1-\epsilon_3)(1+\epsilon_4)^2.$
We define 
\begin{eqnarray}\label{G_h2}
h_2(x)=c_1(1-\epsilon_4x)^{c_2}
\exp\Big\{-c_3 (1-\epsilon_4 x)^{\epsilon_4^{-1}}\Big\}.
\end{eqnarray}
And
\begin{eqnarray}
\int_{0}^{\infty}h_2(x)dx&=&\int_{0}^{\infty}c_1(1-\epsilon_4x)^{c_2}
\exp\Big\{-c_3 (1-\epsilon_4 x)^{\epsilon_4^{-1}}\Big\}dx.\nonumber
\end{eqnarray}
Putting $c_3(1-\epsilon_4 x)^{\epsilon_4^{-1}}=y,$ and $-c_3(1-\epsilon_4 x)^{\epsilon_4^{-1}-1}dx=dy$ we have,
\begin{eqnarray}
\int_{0}^{\infty}h_2(x)dx&=&c_1\left(\frac{1}{c_3}\right)^{\epsilon_4(c_2+1)}
\int_{0}^{c_3}y^{\epsilon_4(c_2+1)-1}\,e^{-y}dy<\infty.\nonumber
\end{eqnarray}
where, $\epsilon_4>\epsilon_2.$\\
Now, we set $I_{F_2}(n,v)=\int_{t_n}^{-v}(g_n(x))^{\beta}\,dx.$ From (\ref{G_h}) for large $n,$ we have $g_n(x)<h_1(x),$ then
\begin{eqnarray}
I_{F_2}(n,v)&<&\int_{-\infty}^{-v}
c_1^\beta (1+\epsilon_4\abs{x})^{\beta c_2}
\exp\Big\{-c_3\beta (1+\epsilon_4 \abs{x})^{\epsilon_4^{-1}}\Big\}dx.
\end{eqnarray}
Putting $c_3\beta(1+\epsilon_4\abs{x})^{\epsilon_4^{-1}}=y,$ and $-c_3\beta(1+\epsilon_4\abs{x})^{\epsilon_4^{-1}-1}dx=dy$ we have,
\begin{eqnarray}
0<I_{F_2}(n,v)&<&c_1^\beta\left(\frac{1}{c_3}\right)^{\epsilon_4(c_2\beta+1)}\int_{c_3\beta(1+\epsilon_4 v)^{\frac{1}{\epsilon_4}}}^{\infty}
y^{\epsilon_4(c_2\beta+1)-1}\,e^{-y}dy.\nonumber
\end{eqnarray}
Therefore,
\begin{eqnarray}
\lim_{v\to\infty}\lim_{n\to\infty} I_{F_2}(n,v)=0.\label{G_IF2}
\end{eqnarray}
From, (\ref{G_IF1}) and (\ref{G_IF2})
\begin{eqnarray}
\lim_{v\to\infty}\lim_{n\to\infty}I_{F}(n,v)=0.\label{G_IF3}
\end{eqnarray}
Set $I_G(n,v)=\int_{-v}^{v}(g_n(x))^\beta\,dx.$ From (\ref{G_h}) and (\ref{G_h2}) for large $n,$ we have $g_n(x)<h_1(x)+h_2(x),$ for and $\int_{-v}^{0}h_1(x)dx+\int_{0}^{v}h_2(x)dx<\infty,$ and using (\ref{s_DenConv_g}) $\lim_{n\to\infty}g_n(x)=\lambda(x),$ for $x\in [-v,v]$ by DCT,
\begin{eqnarray}\label{G_IG}
\lim_{v\to\infty}\lim_{n\to\infty}\int_{-v}^{v}(g_n(x))^\beta\,dx&=&\int_{-\infty}^{\infty}(\lambda(x))^\beta\,dx.
\end{eqnarray}
From, (\ref{G_IE}), (\ref{G_IF3}) and (\ref{G_IG}),
\begin{eqnarray}
\lim_{n\to\infty}H_\beta(g_n)=H_\beta(\lambda).\nonumber
\end{eqnarray}
\end{proof}

\appendix
\section{}
\begin{lem}\label{Lem1}
The R\'{e}nyi entropy of\\
\item[(i)] Fr\'{e}chet law:
\[H_\beta(\phi_{\alpha})=-\log\alpha+\frac{\alpha+1}{\alpha}\log\beta-\frac{1}{1-\beta}\left(\log\beta-\log\Gamma\left(\frac{\alpha+1}{\alpha}(\beta-1)+1\right)\right);\]
where, $\frac{1}{\alpha+1}<\beta.$\\
\item[(ii)] Weibull law:
\[ H_\beta(\psi_{\alpha})=-\log\alpha+\frac{\alpha-1}{\alpha}\log\beta-\frac{1}{1-\beta}\left(\log\beta-\log\Gamma\left(\frac{\alpha-1}{\alpha}(\beta-1)+1\right)\right);\]
where, $\max\left(0,\frac{\beta-1}{\beta}\right)<\alpha,$ for $\beta>0.$\\
\item[(iii)] Gumbel law:
\[H_\beta(\lambda)=\frac{1}{1-\beta}\log\frac{\Gamma(\beta)}{\beta^\beta};\]
where, $\beta> 0.$
\end{lem}
\begin{proof}
(i) The R\'{e}nyi entropy of Fr\'{e}chet distribution is
\begin{eqnarray}
H_\beta(\phi_\alpha)=\frac{1}{1-\beta}\log\int_{0}^{\infty}\left(\alpha x^{-\alpha-1}e^{-x^{-\alpha}}\right)^\beta dx.\nonumber
\end{eqnarray}
Putting, $\beta x^{-\alpha}=u,$ $-\beta\alpha x^{-\alpha-1}dx=du,$
\begin{eqnarray}
H_\beta(\phi_\alpha)&=&\frac{1}{1-\beta}\log\int_{0}^{\infty}\alpha^{\beta-1}u^{(\beta-1)(\frac{\alpha+1}{\alpha})} \beta^{-({(\beta-1)(\frac{\alpha+1}{\alpha})}+1)} e^{-u}du,\nonumber\\
&=&\frac{1}{1-\beta}\left((\beta-1)\log\alpha-\left((\beta-1)\frac{\alpha+1}{\alpha}+1\right)\log\beta+\log\Gamma\left(\frac{\alpha+1}{\alpha}(\beta-1)+1\right)\right).\nonumber
\end{eqnarray}
where, $\frac{\alpha+1}{\alpha}(\beta-1)+1>0,$ so $\frac{1}{\beta}<\alpha+1.$\\
(ii) The R\'{e}nyi entropy of Weibull distribution
\begin{eqnarray}
H_\beta(\psi_\alpha)=\frac{1}{1-\beta}\log\int_{-\infty}^{0}\left(\alpha (-x)^{\alpha-1}e^{-(-x)^{\alpha}}\right)^\beta dx.\nonumber
\end{eqnarray}
Putting, $\beta (-x)^{\alpha}=u,$ $-\beta\alpha (-x)^{\alpha-1}dx=du,$
\begin{eqnarray}
H_\beta(\psi_\alpha)&=&\frac{1}{1-\beta}\log\int_{0}^{\infty}\alpha^{\beta-1}u^{(\beta-1)(\frac{\alpha-1}{\alpha})}\beta^{-((\beta-1)(\frac{\alpha-1}{\alpha})+1)} e^{-u}du,\nonumber\\
&=&\frac{1}{1-\beta}\left((\beta-1)\log\alpha-\left((\beta-1)\frac{\alpha-1}{\alpha}+1\right)\log\beta+\log\Gamma\left(\frac{\alpha-1}{\alpha}(\beta-1)+1\right)\right),\nonumber
\end{eqnarray}
where, $\frac{\alpha-1}{\alpha}(\beta-1)+1> 0.$ If $\beta>1$ then $\alpha> \frac{\beta-1}{\beta},$ and for $\beta<1,$ we have $0<\alpha,$ therefore, $\max\left(0,\frac{\beta-1}{\beta}\right)<\alpha,$ for all $\beta>0.$ \\
(iii) The R\'{e}nyi entropy of Gumbel distribution
\begin{eqnarray}
H_\beta(\lambda)&=&\frac{1}{1-\beta}\log\int_{-\infty}^{\infty}\left(e^{-x}e^{-e^{-x}}\right)^\beta\,dx,\nonumber
\end{eqnarray}
Taking $\frac{u}{\beta}=e^{-x}$ and $\frac{du}{\beta}=-e^{-x}dx$
\begin{eqnarray}
H_\beta(\lambda)&=&\frac{1}{1-\beta}\log\int_{0}^{\infty}u^{\beta-1}e^{-u}\beta^{-\beta}du,\nonumber\\
&=&\frac{1}{1-\beta}\left(\log\Gamma(\beta)-\beta\log\beta\right).\nonumber
\end{eqnarray}
where, $\beta> 0.$
\end{proof}

\begin{lem}\label{R_ent_general1} If $Y=\frac{X-b}{a},$ for $b\in\Real$ and $a>0,$ then the R\'{e}nyi's entropy of $Y$ is given by $$H_\beta(f_Y)=-\log a+H_\beta(f_X).$$
\end{lem}
\begin{proof}
We have $F_Y(y) = \Pr\left(X\leq ay+b\right)=F_X\left(ay+b\right),$ and $f_Y(y) =  a f_X(ay+b),$ so that from (\ref{Renyi_entropy}),
\begin{eqnarray}
H_\beta(f_Y)&=&\frac{1}{1-\beta}\log\int_{-\infty}^{\infty} \left(a f_X(ay+b)\right)^\beta dy \; = \frac{1}{1-\beta}\log\int_{-\infty}^{\infty} f_X^{\beta}(z)a^{\beta-1} dz, \nonumber \\ 
&=&-\log a+H_\beta(f_X).\nonumber
\end{eqnarray}
\end{proof}

\section{}\label{more results}
\begin{defn}(Definition, Page 27, Resnick (1987))\label{Pi vary} A nonegative, nondecreasing function $V(x)$ defined on a semi infinite interval $(z,\infty)$ is $\Pi$ varying (written $V\in\Pi$) if there exist functions $a(t)>0,$ $b(t)\in\Real$ such that for $x>0$
\[\lim_{t\to\infty}\frac{V(tx)-b(t)}{a(t)}=\log x.\]
\end{defn}

\begin{thm}\label{a_rv}(Proposition 0.12, Resnick (1987))
If $V\in\Pi$ with auxiliary function $a(t)$ then $a(\cdot)\in RV_0.$
\end{thm}


\begin{thm}(Proposition 0.8, Resnick (1987))\label{rv}
Suppose $U\in RV_{\rho},$ $\rho\in\Real.$ Take $\epsilon>0.$ Then there exists $t_0$ such that for $x\geq 1$  and $t\geq t_0$
\[(1-\epsilon)x^{\rho-\epsilon}<\frac{U(tx)}{U(t)}<(1+\epsilon)x^{\rho+\epsilon}.\]
\end{thm}

\begin{lem}(Corollary, Page 17, Resnick (1987))
If $L$ is slowly varying iff $L$ can be represented as
\begin{eqnarray}
L(x)=c(x)\exp\Big\{\int_{1}^{x}t^{-1} \epsilon(t)dt\Big\},\label{slow}
\end{eqnarray} for $x>0$ and $\lim_{x\to\infty}c(x)=c$ and $\lim_{t\to\infty}\epsilon(t)=0.$
\end{lem}

\begin{rem}(Remark, Page 19, Resnick (1987))
If $U\in RV_\rho$ then $U$ has representation
\begin{eqnarray}
U(x)=c(x)\exp\Big\{\int_{1}^{x}t^{-1} \rho(t)dt\Big\},\label{kara}
\end{eqnarray}
where, $\lim_{x\to\infty}c(x)=c$ and $\lim_{t\to\infty}\rho(t)=\rho.$
\end{rem}

\begin{thm}(Proposition 2.1, Resnick (1987))\label{Moment}
	Let $F\in\mathcal{D}(G).$
\begin{enumerate}
	\item[(i)] If $G=\Phi_\alpha,$ then $a_n=\left(1/(1-F)\right)^{\leftarrow}(n),$ $b_n=0,$ and if for some integer $0<k<\alpha,$
\[\int_{-\infty}^{0}\mid{x}\mid^kF(dx)<\infty,\]
then $\lim_{n\to\infty}E\left(\dfrac{M_n}{a_n}\right)^k=\int_{0}^{\infty}x^k\Phi_\alpha(dx)=\Gamma\left(1-\frac{k}{\alpha}\right).$
	\item[(ii)] If $G=\Psi_\alpha,$ then $\delta_n=r(F)-\left(1/(1-F)\right)^{\leftarrow}(n),$ $r(F)<\infty,$ and if for some integer $k>0,$
\[\int_{-\infty}^{r(F)}\mid{x}\mid^kF(dx)<\infty,\]
then $\lim_{n\to\infty}E\left(\dfrac{M_n-r(F)}{\delta_n}\right)^k=\int_{-\infty}^{0}x^k\Psi_\alpha(dx)=(-1)^k\Gamma\left(1+\frac{k}{\alpha}\right).$
	\item[(iii)] If $G=\Lambda,$ then $b_n=\left(1/(1-F)\right)^{\leftarrow}(n),$ $a_n=u(b_n),$ and if for some integer $k>0,$
\[\int_{-\infty}^{0}\mid{x}\mid^kF(dx)<\infty,\]
then $\lim_{n\to\infty}E\left(\dfrac{M_n-b_n}{a_n}\right)^k=\int_{-\infty}^{\infty}x^k\Lambda(dx)=(-1)^k\Gamma^{(k)}(1),$
where $\Gamma^{(k)}(1)$ is the k-th derivative of the Gamma function at $x=1.$
\end{enumerate}
\end{thm}

\begin{thm}(Proposition 1.15 and 1.16, Resnick (1987))\label{thm_von}
\item[(i)] Suppose that distribution function $F$ is absolutely continuous with density $f$ which is eventually positive. \\
\item[(a)] If for some $\alpha>0$
\begin{eqnarray}\label{von_F}
\lim_{x\to\infty}\dfrac{xf(x)}{\overline{F}(x)}=\alpha,
\end{eqnarray}
then $F\in\D(\Phi_\alpha).$\\
\item[(b)] If $f$ is nonincreasing and $F\in\D(\Phi_\alpha)$ then (\ref{von_F}) holds.
\item[(ii)] Suppose $F$ has right endpoint $r(F)$ finite and density $f$ positive in a left neighbourhood of $r(F).$  \\
\item[(a)] If for some $\alpha>0$
\begin{eqnarray}\label{von_W}
\lim_{x\to r(F)}\dfrac{(r(F)-x)f(x)}{\overline{F}(x)}=\alpha
\end{eqnarray}
then $F\in\D(\Psi_\alpha).$\\
\item[(b)] If $f$ is nonincreasing and $F\in\D(\Psi_\alpha)$ then (\ref{von_W}) holds.
\end{thm}

\begin{thm}\label{thm_g_von}(Proposition 1.17, Resnick (1987)) Let  $F$ be absolutely continuous in a left neighborhood of $r(F)$ with density $f.$ If 
\begin{eqnarray} \label{lem3_appendix} \lim_{x\uparrow r(F)}f(x)\int_{x}^{r(F)}\overline{F}(t)dt/\overline{F}(x)^2=1,\end{eqnarray} 
 then $F\in\D(\Lambda).$ In this case we may take,
\[u(x)=\int_{x}^{r(F)}\overline{F}(t)dt/\overline{F}(x), \;b_n=F^{\leftarrow}(1-1/n),\;\;a_n=u(b_n).\]
\end{thm}

\begin{thm} (Theorem 2.5, Resnick (1987)),\label{gn_conv}
Suppose that $F$ is absolutely continuous with pdf $f.$ If $F\in\D(G)$ and
\begin{itemize}
\item[(i)] $G=\Phi_\alpha,$ then $g_n(x)\to \phi_\alpha(x)$ locally uniformly on $(0,\infty)$ iff (\ref{von_F}) holds;
\item[(ii)] $G=\Psi_\alpha,$ then $g_n(x)\to \psi_\alpha(x)$ locally uniformly on $(-\infty,0)$ iff (\ref{von_W}) holds;
\item[(iii)] $G=\Lambda,$ then $g_n(x)\to \lambda(x)$ locally uniformly on $\Real$ iff (\ref{lem3_appendix}) holds.
\end{itemize}
\end{thm}

\begin{thm}\label{Lem4_appendix}(Lemma 2, De Haan and Resnick (1982))
Suppose $F\in \D(\Lambda)$ with auxiliary function $u$ and $\epsilon>0.$ There exists a $t_0$ such that for $x\geq 0,$ $t\geq t_0$
\begin{eqnarray}
(1-\epsilon)\left[\frac{-\log F(t)}{-\log F(t+xu(t))}\right]^{-\epsilon}\leq \frac{u(t+xu(t))}{u(t)}\leq (1+\epsilon)\left[\frac{-\log F(t)}{-\log F(t+xu(t))}\right]^{\epsilon},\nonumber
\end{eqnarray}
and for $x<0,$ $t+xu(t)\geq t_0$
\begin{eqnarray}
(1-\epsilon)\left[\frac{-\log F(t+xu(t))}{-\log F(t)}\right]^{-\epsilon}\leq \frac{u(t+xu(t))}{u(t)}\leq (1+\epsilon)\left[\frac{-\log F(t+xu(t))}{-\log F(t)}\right]^{\epsilon}.\nonumber
\end{eqnarray}
\end{thm}

\begin{thm}\label{Lem5_appendix}(Corollary, Balkema and De Haan (1972))
A distribution function $F\in\D(\Lambda)$ if and only if there exist a positive function $c$ satisfying $\lim_{x\to r(F)}c(x)=1$ and a positive differentiable function $u(t)$ satisfying $\lim_{x\to r(F)}u'(x)=0$ such that
\[\bar{F}(x)=c(x)\exp\Big\{-\int_{-\infty}^{x}\frac{dt}{u(t)}\Big\}\;\;\text{for }x<r(F).\]
\end{thm}

\end{document}